\documentclass[conference]{IEEEtran}
\usepackage[T1]{fontenc} 
\usepackage{times}
\usepackage{amsmath}
\usepackage{cases}
\usepackage{amssymb}
\usepackage{amsthm}
\usepackage{color}
\usepackage{algorithm}
\usepackage{graphicx}
\usepackage{subfig}
\usepackage{multirow}
\usepackage{url}
\usepackage[bookmarks=false,colorlinks=false,pdfborder={0 0 0}]{hyperref}
\usepackage{cite}
\usepackage{bm}
\usepackage{arydshln}
\usepackage{mathtools}
\usepackage{microtype}
\usepackage{algorithm}
\usepackage{algorithmic}
\usepackage{geometry}
\usepackage[figuresright]{rotating}
\usepackage{threeparttable}
\usepackage{booktabs}

\newcounter{MYtempeqncnt}

\newtheorem{theorem}{Theorem}

\newtheorem{corollary}[theorem]{Corollary}
\newtheorem{example}[theorem]{Example}
\newtheorem{claim}[theorem]{Claim}

\long\def\symbolfootnote[#1]#2{\begingroup
\def\thefootnote{\fnsymbol{footnote}}\footnote[#1]{#2}\endgroup}
\renewcommand{\paragraph}[1]{{\bf #1}}
\title{Accelerating Data Access for Single Node in Distributed Storage Systems via MDS Codes}

\author{Hao Shi$^\dagger$, Zhengyi Jiang$^\dagger$,  Zhongyi Huang$^\dagger$,
Linqi Song$^{\ast}$$^{\divideontimes}$, and Hanxu Hou$^\ddagger$$^{\ast}$$^{\star}$\\
$^\dagger$ Department of Mathematics Sciences, Tsinghua University, Beijing, China \\
$^\ddagger$ School of Computer Science and Technology, Dongguan University of Technology\\
$^{\ast}$ Department of Computer Science, City University of Hong Kong\\
$^{\divideontimes}$ City University of Hong Kong Shenzhen Research Institute
}

\begin{document}
\let\emph\textit
\maketitle
\pagestyle{empty}  
\thispagestyle{empty} 
\symbolfootnote[0]{$^\ddagger$:  Corresponding author. 
}

\begin{abstract}
    Maximum distance separable (MDS) array codes are widely employed in modern
distributed storage systems to provide high data reliability with small storage overhead. 
Compared with the data
access latency of the entire file, the data access latency of a single node in a distributed 
storage system is equally important. In this paper, we propose two algorithms to effectively 
reduce the data access latency on a single node in different scenarios for MDS codes. We show theoretically that 
our algorithms have an expected reduction ratio of $\frac{(n-k)(n-k+1)}{n(n+1)}$ and $\frac{n-k}{n}$ for 
the data access latency of a single node when it obeys uniform distribution and shifted-exponential distribution, 
respectively, where $n$ and $k$ are the numbers of all nodes and the number of data nodes respectively. In the 
worst-case analysis, we show that our algorithms have a reduction ratio of more than $60\%$ when $(n,k)=(3,2)$. 
Furthermore, in simulation experiments, we use the Monte Carlo simulation algorithm to demonstrate 
less data access latency compared with the baseline algorithm.
\end{abstract}


\section{Introduction}
Maximum distance separable (MDS) array codes are widely employed in modern
distributed storage systems to provide high data reliability with small storage overhead. An $(n,k,m)$ MDS array code encodes $km$
{\em data symbols} into $nm$ {\em coded symbols} that are equally stored in $n$ nodes,
where each node stores $m$ symbols. We call the number of symbols stored in each node
as the sub-packetization level. The $(n,k,m)$ MDS array codes satisfy the {\em MDS property},
that is, any $k$ out of $n$ nodes can retrieve all $km$ data symbols. The codes are referred
to as {\em systematic codes} if the $km$ data symbols are included in the $nm$ coded symbols. Reed-Solomon (RS) codes \cite{reed1960} are typical MDS array codes
with $m = 1$. In this paper, we consider systematic
MDS array codes that contain $k$ {\em data nodes} which store the $km$ data symbols
and $r=n-k$ {\em parity nodes} which store the $rm$ parity symbols.

In modern distributed storage systems, there are indicators worth optimizing, such as repair bandwidth \cite{rashmi2011,tamo2013,hou2016,2017Explicit,li2018,2018A,hou2019a,hou2019b}, update bandwidth\cite{9514857,10437570,10619528}, 
I/O cost \cite{8849700, 8437865, 2024formula}, etc., but researchers often overlook an important indicator, which is the latency of data access from storage systems. 
In distributed storage systems, low latency is important for humans \cite{vulimiri2012more, ahmad2024efficient}. Even slightly higher web page load time can significantly 
reduce visits from users and revenue, as demonstrated by several sites \cite{2009web}. However, achieving consistent low latency is challenging. 
Modern applications are highly distributed \cite{alizadeh2010data, 2010web}, and likely to get more so as cloud computing separates users from their data and computation.

The simple replication method can slightly reduce the latency of data access, but it requires a huge amount of storage space \cite{acharya2024existence}. Given that total replication of files at storage nodes is impractical, several schemes based on partial replication have been proposed in the literature, where each storage node stores only a subset of the data (see \cite{cadambe2023brief} and references therein).
Some previous work \cite{huang2012codes, joshi2012coding, shah2014mds, lee2017mds, joshi2014delay} has used the MDS codes to study data access latency. They aim to divide a file into $k$ nodes for storage
and encode it into $n$ nodes. We only need to access the information in any $k$ nodes to get 
the entire file. This technology can effectively reduce the latency of accessing the entire file data. It uses knowledge such as task 
scheduling and queuing theory to achieve some good results. However, in real-world scenarios, we don’t need to access all the data in a file in most cases, 
but only a part of it. The current approach to this problem is to directly access the needed data of the node without any other conversion. 
The data access latency of each node is often different, depending on the physical distance and current status of the node. Once the access 
latency of the node we need is very high, this will cause the feedback latency of the service to become very high, affecting normal data usage.

 In this paper, we study the data access latency of a single data node. Based on the MDS codes, we propose two algorithms to reduce 
 the data access latency of a single data node in the low-load system. The two algorithms are suitable for two scenarios: the access latency of each node is known and unknown. 
 Note that the difference between our work and previous studies is that previous studies focus on the data access latency of the entire file, while we focus on the 
 access latency of a single data node. In actual scenarios, we do not need all the data of the file in most cases, but only a part of the data of the file. 
 The current industry approach to this problem is to directly access the data of the node without any other conversion.

 In the rest of the paper, we first propose two algorithms based on the MDS property of data storage to deal with two different scenarios to reduce the access latency of a single data node effectively. Then, we extend a more general Shifted-Exponential distribution based on the Shifted-Exponential distribution in \cite{lee2017speeding, liang2014tofec} to evaluate our algorithms. Finally, we use the Monte Carlo simulation algorithm to show that our algorithms have less data access latency than the baseline algorithm.

\section{Data Access via MDS Property}

In this section, we present a new data access method via the MDS property of 
coded data in the distributed storage systems.

Consider the $m\times n$ MDS code array such that: (i) the symbol $a_{i,j}$ in 
row $j=1,2,\ldots,m$ and column $i=1,2,\ldots,k$ is data symbol and the symbol 
$\{f_{i,j}(\boldsymbol{a}_1,\ldots,\boldsymbol{a}_m)\}$ in row $j=1,2,\ldots,m$ and column 
$k+i=k+1,k+2,\ldots,k+r$ is parity symbol, where $\boldsymbol{a}_j:=(a_{1,j},a_{2,j},
\ldots,a_{k,j})$ represents the $k$ data symbols;  
(ii) the symbol $a_{i,j}$ 
in row $j=1,2,\ldots,m$ is stored in node $i$ where $i = 1,2,\ldots, k$ and the symbol 
$\{f_{i,j}(\boldsymbol{a}_1,\ldots,\boldsymbol{a}_m)\}$ in row 
$j=1,2,\ldots,m$ is stored in node $k+i$, where $i = 1,2,\ldots, r$.

\begin{figure}[htbp]
    \centering
    \includegraphics[width=0.48\textwidth]{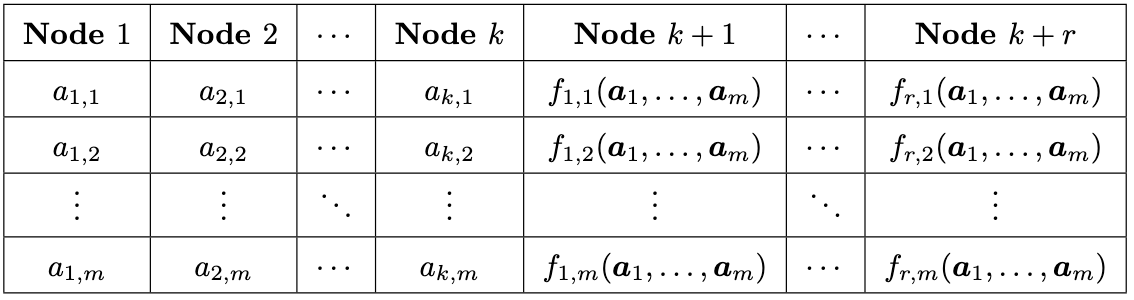}
    \caption{The structure of $(n, k)$ MDS systematic codes.}
    \label{fig.1.0}
\end{figure}


Suppose that we want to access the data stored in a certain data node, 
we can directly access the data symbols from that node 
and then get the data we need. 
In the rest of this paper, we call the above direct data Access method the Directly Access (DA) method.
However, if the latency of data access from that node happens to be very high, we will experience a potentially long period 
between making the request to extract the data and getting the data, which is inefficient.
For this reason, we propose our solution, which we will start with an example.

\begin{example}
    Suppose a set of data consists of four symbols $a_{1,1}, a_{1,2}, a_{2,1}, a_{2,2}$ and two 
parity symbols $a_{1,1}+ a_{2,1}, a_{1,2}+a_{2,2}$ consisting of $(3,2,2)$ MDS code, which are stored in 
nodes 1,2,3 respectively, and the latency of the data that we want to access for the three 
nodes is set to $x_1 = 50, x_2 = 100, x_3 = 50$ respectively. At a certain 
moment, we need to use $a_{2,1}, a_{2,2}$ two symbols immediately, directly access the data of node 
2, then we have a total latency of $x_2 = 100$; in addition, we can turn the idea to 
access the data of node 1 and node 3  $a_{1,1}, a_{1,2}, a_{1,1} + a_{2,1}, a_{1,2} + a_{2,2}$ four symbols, 
we can access the data in node 1 and node 3 in parallel, so we only 
need a total latency of $\max(x_1, x_3) = 50$, and then do some simple calculations
\begin{eqnarray*}
a_{1,1} \oplus (a_{1,1}\oplus a_{2,1}) = a_{2,1} \\
a_{1,2} \oplus (a_{1,2}\oplus a_{2,2}) = a_{2,2}
\end{eqnarray*}
This gives us the data $a_{2,1},a_{2,2}$ and we can reduce the access latency of $\frac{100-50}{100}*100\% = 50 \%$.
\end{example}

\begin{figure}[htbp]
    \centering
    \includegraphics[width=0.35\textwidth]{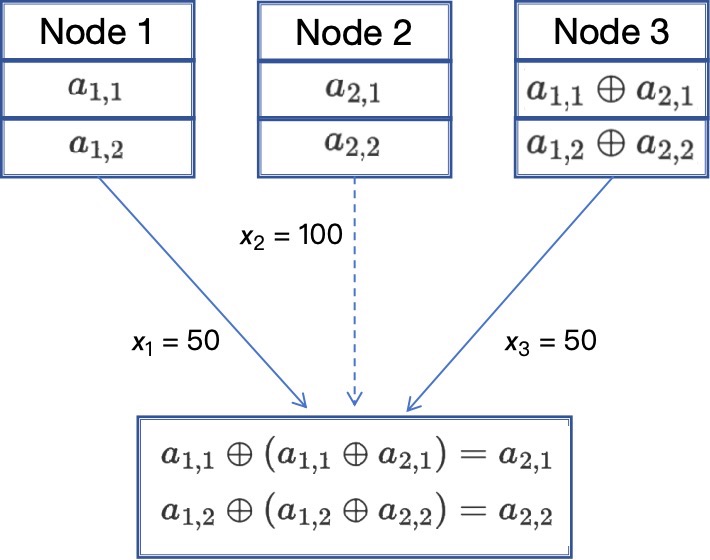}
    \caption{\textbf{An example for data access latency in (3, 2, 2) MDS codes.} When we want to get the data stored in node 2, we have two choices: i) Directly access data in node 2 (the dotted line in the figure); ii) If the access latency of node 2 is larger than that of other two nodes, we can access data stored in node 1 and node 3 and then through some calculations we can get the data stored in node 2 (the solid line in the figure).}
    \label{fig.2.0}
\end{figure}

Now, we discuss the new idea to reduce the data access latency through the MDS 
property of the code in the general case.
For a general distributed storage system containing $n$ nodes, which are stored in the 
$(n, k, m)$ MDS array code shown in Fig. 1. 
For $i= 1,2, \ldots, n$, assume that the access latency of the data of node $i$ is $x_i$.

Suppose that we need to access the data in data node 
$t$ (we also call this a data access request $Q_t$), where $1 \leq t \leq k$. 
On the one hand, we can directly access the data from node $t$, which is the most conventional idea. In this way, the access latency is $x_t$. On the other hand, note that the original data can be recovered from all the symbols of any $k$ nodes because of the $(n,k)$ MDS property. Therefore, if we 
access the symbols from $k$ nodes $\{j_v\}_{v=1}^{k}$ in parallel, we can get 
all the data through some calculations, which of course includes the data of node $t$, and the data access latency of this 
way is $\max(\{x_{j_v}\}_{v=1}^{k})$. Of course, we should choose $k$ nodes to make $\max(\{x_{j_v}\}_{v=1}^{k})$ as small as possible.

If we can know the relative scale of data access latency for each node in advance, for example,
$x_{j_1}\leq x_{j_2} \leq \cdots \leq x_{j_n}$, in which $(j_1, j_2, \ldots, j_n)$ is a permutation of $\{1,2,\ldots,n\}$,
we only consider the data access latency of the nodes in the set $\{x_{j_s}\}_{s=1}^k$. 
If node $t$ happens to be in $\{x_{j_s}\}_{s=1}^k$, we can directly access the data in node $t$, otherwise, we only need $x_{j_k}$ for the latency of accessing all the symbols from $k$ nodes $\{j_v\}_{v=1}^{k}$, 
and we can get the $m$ data symbols in node $t$ by some further calculations. 
The algorithm can be found in the pseudo-code Alg. \ref{alg:1}.

\begin{algorithm}
    \caption{Accelerated Access with Known Latency Algorithm}
    \begin{algorithmic}[1]
    \REQUIRE Data access request $Q_t$, set of $k$ nodes with the lowest latency $D$;
    \ENSURE Data symbols to be accessed $a_{t,1}, a_{t,2}, \ldots, a_{t,m}$;
    
    \STATE Receive data access request $Q_t$;
    \IF{$t \in D$}
        \STATE Server accesses the distributed storage system to access data from node $t$;
        \RETURN Data from node $t$.
    \ENDIF
    \STATE Server accesses the distributed storage system to access all data symbols concurrently (denoted by $A=\{a_{i,j}\}_{i\in D,j=1,2,\ldots,m}$) from $k$ nodes in $D$;
    \RETURN Compute ($A$, $t$).
    \end{algorithmic}
    \label{alg:1}
\end{algorithm}

\begin{algorithm}
    \caption{Compute $(A,t)$}
    \begin{algorithmic}[1]
    \REQUIRE Data symbols set $A$ accessed from $k$ nodes by the server, data node $t$ needed;
    \ENSURE Data symbols in node $t$;
    
    \STATE List equations for the $k$ nodes based on the encoding algorithm;
    \STATE Retrieve the data symbols $a_{t,1}, a_{t,2}, \ldots, a_{t,m}$ from node $t$ by $(n,k)$ MDS property;
    \RETURN $a_{t,1}, a_{t,2}, \ldots, a_{t,m}$.
    \end{algorithmic}
    \label{alg:2}
\end{algorithm}

In most cases, the relative scale of data access latency for each node cannot be known in advance, we can still use a similar idea. 
For a data access request $Q_t$, where $t\in\{1,2,\ldots,k\}$. In the first stage, we simultaneously start the access process of the symbols of all $n$ nodes. Note that the access latency of each node may be different, once we access all $mk$ symbols of certain $k$ nodes, 
we stop the access process of the remaining nodes.
In the second stage, we can further get the needed data symbols in node $t$ based on all the $mk$ symbols that have been accessed via the $(n,k)$ MDS property. 
It is worth noting that, all $m$ symbols of node $t$ may have been accessed in the first stage.
Once this situation is detected, 
we can terminate the subsequent process and extract the data symbols of node $t$ directly to obtain the required data without the second stage. The specific algorithm can be found in the pseudo-code Alg. \ref{alg:3}.

\begin{algorithm}
    \caption{Accelerated Access with Unknown Latency Algorithm}
    \begin{algorithmic}[1]
    \REQUIRE Data access request $Q_t$;
    \ENSURE Data to be accessed $a_{t,1}, a_{t,2}, \ldots, a_{t,m}$;
    
    \STATE Receive data access request $Q_t$;
    \STATE Set the number of nodes that have finished access to $flag = 0$;
    \STATE Server accesses the distributed storage system to concurrently access data from $n$ nodes;
    
    \WHILE{a node finishes accessing}
        \IF{the data is from node $t$}
            \RETURN Data from node $t$.
        \ELSE
            \STATE $flag += 1$;
        \ENDIF
        \IF{$flag == k$}
            \RETURN Compute($A$, $t$).
        \ENDIF
    \ENDWHILE
    \end{algorithmic}
    \label{alg:3}
\end{algorithm}

In this paper, we denote the two proposed MDS property-based accelerated data access algorithms as Accelerated Access with 
Known Latency (AAKL) (Alg.~\ref{alg:1}) and Accelerated Access with UnKnown Latency (AAUL) Algorithm (Alg.~\ref{alg:3}). Compared with the AAUL algorithm, the AAKL algorithm knows the specific distribution of data download delays of each node in advance, so it can accurately determine which nodes to download data from. It does not need to download data from all nodes, and is more efficient in terms of resource usage.

\section{Theoretical Analysis of Two Algorithms}
In this section, we theoretically show that our AAKL and AAUL algorithms have less latency than that of the DA method.

For generality, assume that the data access latency of $n$ nodes denoted by $X_1, X_2,\ldots, X_n$ that are $n$ 
independent and identically distributed random variables, all obeying a given distribution. 
When we need to access data from data node $t$, where $1\leq t \leq k$, the data access latency of the DA algorithm is $X_t$, and it is easy to find out that regardless of knowing the data access latency for each node, the data access latency of our AAKL and AAUL algorithms are the same, i.e., 
$\min(X_t, \max(mink(\{X_i\}_{i \neq t})))$, where $mink(S)$ denotes the subset of the smallest $k$ elements in $S$ for any set $S$.
Note that 
$$X_t \geq \min(X_t, \max(mink(\{X_i\}_{i\neq t}))).$$
Then we have the following claim without proof.

\begin{claim}
    The data access latency of AAKL or AAUL algorithm is no more than that of the DA algorithm. 
\end{claim}

Let $Y_{t,1} = X_t, Y_{t,2} = \min(X_t, \max(mink(\{X_i\}_{i\neq t})))$, for each $t\in\{1,2,\ldots,n\}$. For any random variable $X$, we denote $f_X(\cdot)$ as the probability 
density function (PDF) of $X$ and $F_X(\cdot)$ as the cumulative probability distribution function 
(CDF) of $X$. In order to better analyze the performance 
of the AAKL and AAUL algorithms theoretically, the following theorem gives the CDF and PDF of $Y_{t,1},Y_{t,2}$.

\begin{theorem}
    For any $t\in\{1,2,\ldots,n\}$, denote that the CDF and PDF of the random variable $\{X_i\}_{i\neq t}$ are $F_{X_{t,0}}(y), f_{X_{t,0}}(y)$, and the CDF and PDF of the 
    random variable $X_t$ are $F_{X_t}(y), f_{X_t}(y)$, we can obtain the CDF and PDF of $Y_{t,1}$ and $Y_{t,2}$ as Eq. \eqref{eq:6789}.
\end{theorem}

\begin{figure*}[!t]
    \normalsize
    \setcounter{MYtempeqncnt}{\value{equation}}
    \setcounter{equation}{5}
    \hrulefill
    \begin{eqnarray}\label{eq:6789}
        F_{Y_{t,1}}(y) &=& F_{X_t}(y) \\\nonumber
        f_{Y_{t,1}}(y) &=& f_{X_t}(y) \\\nonumber
        F_{Y_{t,2}}(y) &=& F_{X_t}(y) + \Big(1 - F_{X_t}(y)\Big) \cdot \sum_{i=k}^{n-1}\binom{n-1}{i}F_{X_{t,0}}(y)^i(1-F_{X_{t,0}}(y))^{n-1-i} \\\nonumber
        f_{Y_{t,2}}(y) &=& f_{X_t}(y) - f_{X_t}(y)\cdot \sum_{i=k}^{n-1}\binom{n-1}{i}F_{X_{t,0}}(y)^i(1-F_{X_{t,0}}(y))^{n-1-i} + \Big(1 - F_{X_t}(y)\Big)\\\nonumber
        &&\cdot \sum_{i=k}^{n-1}\Big[C^{i}_{n-1}if_X(y)F_{X_{t,0}}(y)^{i-1}(1-F_{X_{t,0}}(y))^{n-1-i} - C^{i}_{n-1}(n-1-i)f_X(y)F_{X_{t,0}}(y)^{i}(1-F_{X_{t,0}}(y))^{n-2-i}\Big].
    \end{eqnarray}
    \setcounter{equation}{\value{MYtempeqncnt}}

    \vspace*{4pt}
\end{figure*}

\begin{proof}
    We can easily get that $F_{Y_{t,1}}(y) = F_{X_t}(y)$ and $f_{Y_{t,1}}(y) = f_{X_t}(y)$ because of $Y_{t,1} = X_t$.

    We can derive the explicit form of $F_{Y_{t,2}}(y)$ as follows.
    \begin{eqnarray*}
        &&F_{Y_{t,2}}(y)\\ &=&  P(Y_{t,2} \leq y)\\ 
        &=&  P(\min(X_t, \max(mink(\{X_i\}_{i\neq t})) \leq y)\\
        &=& P(X_t \leq y \vee \max(mink(\{X_i\}_{i\neq t})) \leq y) \\
        &=& 1 - P(X_t > y ) \cdot P(\max(mink(\{X_i\}_{i\neq t})) > y) \\ 
        &=& 1- (1 - F_{X_t}(y))\cdot(1-F_{\max(mink(\{X_i\}_{i\neq t}))}(y)).
    \end{eqnarray*}
    Notice that
    \begin{eqnarray*}
        &&F_{\max(mink(\{X_i\}_{i\neq t}))}(y)\\ &=& P(\max(mink(\{X_i\}_{i\neq t})) \leq y)\\ 
        &=& P(mink(\{X_i\}_{i\neq t}) \leq y) \\ 
        &=& \sum_{i=k}^{n-1}\binom{n-1}{i}P(X \leq y) ^ {i}(1-P(X\leq y))^{n-1-i} \\ 
        &=& \sum_{i=k}^{n-1}\binom{n-1}{i}F_{X_{t,0}}(y)^i(1-F_{X_{t,0}}(y))^{n-1-i}.
    \end{eqnarray*}
    So we can get that
    \begin{eqnarray*}
        &&F_{Y_{t,2}}(y) \\&=& 1- (1 - F_{X_t}(y))\cdot(1-F_{\max(mink(\{X_i\}_{i\neq t}))}(y)) \\
        &=& 1- (1 - F_{X_t}(y))\cdot\\&&(1- \sum_{i=k}^{n-1}\binom{n-1}{i}F_{X_{t,0}}(y)^i(1-F_{X_{t,0}}(y))^{n-1-i}) \\
        &=& F_{X_t}(y) + \sum_{i=k}^{n-1}\binom{n-1}{i}F_{X_{t,0}}(y)^i(1-F_{X_{t,0}}(y))^{n-1-i} \\&&- \sum_{i=k}^{n-1}C^{i}_{n-1}F_{X_{t,0}}(y)^i(1-F_{X_{t,0}}(y))^{n-1-i}\cdot F_{X_t}(y) \\
        &=&F_{X_t}(y) + \Big(1 - F_{X_t}(y)\Big) \cdot \sum_{i=k}^{n-1}\binom{n-1}{i}\\&&F_{X_{t,0}}(y)^i(1-F_{X_{t,0}}(y))^{n-1-i}.
    \end{eqnarray*}
\end{proof}

\subsection{Expectation Case Analysis}

For $t\in\{1,2,\ldots,n\}$, we denote the mathematical expectation of the random variables 
$Y_{t,1}$ and $Y_{t,2}$ as $E(Y_{t,1})$ and $E(Y_{t,2})$, respectively.

\begin{corollary}\label{lemma1}
    If the random variables $\{X_i\}_{i=1}^n$ obey the uniform distribution $(X_i\sim U[0, T], i =1,2,\ldots,n),$ 
    for any $(n,k,m)$ MDS array code, the expectation of the data access latency
    for AAKL (or AAUL) algorithm is 
    $$E[Y_{t,2}] = \frac{T}{2} - \frac{(n-k)(n-k+1)T}{2n(n+1)}.$$
\end{corollary}
\begin{proof}
Recall that the random variable $\{X_i\}_{i=1}^n$ obeys the uniform distribution $(X_i\sim U[0, T], i =1,2,\ldots,n)$, 
    we can know that
    \begin{eqnarray*}
        &&F_{X_{t,0}}(y) = F_{X_t}(y) = \frac{y}{T}, \qquad y \in [0, T], \\
        &&f_{X_{t,0}}(y) = f_{X_t}(y) = \frac{1}{T}, \qquad y \in [0, T]. 
    \end{eqnarray*}
Together the two equations above and Eq. \eqref{eq:6789}, we get
    \begin{eqnarray*}
        &&f_{Y_{t,2}}(y) = \frac{1}{T} + \sum_{i=k}^{n-1} \binom{n-1}{i} ( \frac{i y^{i-1}}{T^i} \left(1 - \frac{y}{T}\right)^{n-1-i} \\
        &&- \frac{(n-1-i) y^i}{T^{i+1}} \left(1 - \frac{y}{T}\right)^{n-2-i} ) \\
        &&- \sum_{i=k}^{n-1} \binom{n-1}{i} ( \frac{(i+1) y^i}{T^{i+1}} \left(1 - \frac{y}{T}\right)^{n-1-i}\\
        && - \frac{(n-1-i) y^{i+1}}{T^{i+2}} \left(1 - \frac{y}{T}\right)^{n-2-i} ), \qquad y \in [0, T].
    \end{eqnarray*}
    Then we can get
    \begin{align*}
        E[Y_{t,2}] &= \int_{-\infty}^{+\infty}yf_{Y_{t,2}}(y)dy = \frac{T}{2} - \frac{(n-k)(n-k+1)T}{2n(n+1)}.
    \end{align*}

\end{proof}
\begin{corollary}\label{coro5}
    If the random variables $\{X_i\}_{i=1}^n$ obey the uniform distribution $(X_i\sim U[0, T], i =1,2,\ldots,n)$, our AAKL (or AAUL) algorithm can reduce 
    the expectation of the data access latency compared to the DA algorithm by the rate of
    $$\Gamma_{U} = \frac{(n-k)(n-k+1)}{n(n+1)}.$$
\end{corollary}
\begin{proof}
    From Corollary \ref{lemma1}, we have $E[Y_{t,1}] = \frac{T}{2}, E[Y_{t,2}] = \frac{T}{2} - \frac{(n-k)(n-k+1)T}{2n(n+1)}$. Then we can obtain that 
    \begin{eqnarray*}
        \Gamma_{U} &=& \frac{E[Y_{t,1}] - E[Y_{t,2}]}{E[Y_{t,1}]}
        = \frac{(n-k)(n-k+1)}{n(n+1)}.
    \end{eqnarray*}

\end{proof}

Denote the code rate of $(n,k)$ MDS code as $c = \frac{k}{n}$. From Corollary \ref{coro5}, we can get
$
    \Gamma_{U} = \frac{(n-k)(n-k+1)}{n(n+1)} = \frac{(1-c)(1-c+\frac{1}{n})}{(1+\frac{1}{n})}.
$
If $n \to +\infty$, we have $\Gamma_{U} = (1-c)^2$, $\Gamma_{U}$ is only related to the code rate of the MDS code. This indicates that when the code rate $c$ is smaller, $\Gamma_U$ is larger, then the benefit of using AAKL (AUKL) compared to the AD algorithm is greater, and it is a quadratic growth relationship.

In real distributed storage systems, the data access latency of each node is often not uniformly distributed, but has an 
initial threshold $s$. The data access latency of each node must be greater than $s$, and it shows an exponential distribution 
trend starting from $s$. Therefore, we consider a new distribution based on the exponential distribution, Shifted-Exponential 
distribution (Shifted-Exp($\lambda, s$)). It is an exponential distribution shifted to the right by $s$ units, which is more 
consistent with the distribution of data access latency of each node in practical scenarios. It is the closest distribution to the actual data download latency distribution of each node \cite{lee2017speeding}. It shows that each node has a minimum threshold $s$ for data download latency, and each different node will have a slightly different response time due to different loads and communication conditions. In particular, when $s=0$, it 
degenerates into an exponential distribution; when $s=1$, it is the distribution shown in \cite{lee2017speeding}. Next, we analyze the data access latency  
of our algorithm based on Shifted-Exponential distribution.

\begin{corollary}\label{lemma2}
    If the random variable $\{X_i\}_{i=1}^n$ obeys the Shifted-Exponential distribution $(X_i\sim $ Shifted-Exp$(\lambda, s ), i =1,2,\ldots,n),$ 
    for any $(n,k,m)$ MDS array code, the expectation of the data access latency
    for AAKL (or AAUL) algorithm is 
    $$E[Y_{t,2}] = s+\frac{k}{\lambda n}.$$
\end{corollary}
\begin{proof}
    Because the random variable $\{X_i\}_{i=1}^n$ obeys the uniform distribution $(X_i\sim $ Shifted-Exp$(\lambda, s ), i =1,2,\ldots,n),$ 
    we can know that
    \begin{equation}
        F_{X_{t,0}}(y) = F_{X_{t}}(y) = 
        \begin{cases}
            &1 - e^{-\lambda (y-s)},  \hfill y \geq s \\
            &0 ,  \hfill y < s\\
        \end{cases};
    \end{equation}
    \begin{eqnarray*}
        f_{X_{t,0}}(y) = f_{X_t}(y) = 
        \begin{cases}
            &\lambda e^{-\lambda (y-s)}, \hfill y \geq s \\
            & 0,  \hfill y < s
        \end{cases}.
    \end{eqnarray*}
    Together the two equations above and Eq. \eqref{eq:6789}, we further get (Due to space limitation, we omit the tedious calculation details here)   
    $$E[Y_{t,2}] = \int_{-\infty}^{+\infty}yf_{Y_{t,2}}(y)dy=s+\frac{k}{\lambda n}.$$

\end{proof}
\begin{corollary}
    If the random variable $\{X_i\}_{i=1}^n$ obeys the Shifted-Exponential distribution $(X_i\sim $ Shifted-Exp$(\lambda, s ), i =1,2,\ldots,n),$ 
    for any $(n,k,m)$ MDS array code, our AAKL (or AAUL) algorithm can reduce 
    the expectation of the data access latency compared to the DA algorithm by the rate of
    $$\Gamma_{SE} = \frac{1}{s\lambda + 1}\cdot\frac{n-k}{n}.$$
\end{corollary}
\begin{proof}
    From Corollary \ref{lemma2}, we have $E[Y_{t,1}] = s+\frac{1}{\lambda}, E[Y_{t,2}] = s+\frac{k}{\lambda n}$.
    Then we can obtain that 
    \begin{eqnarray*}
        \Gamma_{SE} &=& \frac{E[Y_{t,1}] - E[Y_{t,2}]}{E[Y_{t,1}]} = \frac{1}{s\lambda + 1}\cdot\frac{n-k}{n}.
    \end{eqnarray*}
\end{proof}
From Corollary 7, we can get
$$
    \Gamma_{SE} = \frac{1}{s\lambda + 1}\cdot\frac{n-k}{n} = \frac{1-c}{s\lambda + 1},
$$
which is related to the code rate $c$ of the MDS code and the parameters of distribution $s$ and $\lambda$. When the code rate $c$ is smaller, $\Gamma_{SE}$ is larger, then the benefit of using AAKL (AUKL) compared to the AD algorithm is greater, and it is a linear growth relationship.

\subsection{Worst Case Analysis}

For convenience, the example in Example 1 is used for the worst-case analysis. This example was selected because our findings indicate that access latency does not significantly depend on the specific values of $n$ and $k$. Instead, it is closely related to the code rate. 

Assume that $X_1$, $X_3$ follow uniform distribution (U[0, T]), and $X_2$ is a constant $T$, then $Y_{t,1} = T$, $Y_{t,2} = \min(T, \max(X_1, X_3)) = \max(X_1, X_3)$.
The PDF of $Y_{t,2}$ is
$$
f_{Y_{t,2}}(y) = 2\frac{y}{T^2}, \qquad y \in [0, T]
$$
The CDF of $Y_{t,2}$ is
$$
F_{Y_{t,2}}(y) = \frac{y^2}{T^2}, \qquad y \in [0, T].
$$
Calculating the expectation of the random variable $Y_{t,2}$ yields
\begin{eqnarray*}
E[Y_{t,2}] 
= \int^{T}_{0}yf_{Y_{t,2}}(y)dy 
= \frac{2}{3}T.
\end{eqnarray*}
At this time, 
\begin{eqnarray*}
    \Gamma_{U}^{Worst} = \frac{E[Y_{t,1}]-E[Y_{t,2}]}{E_{Y_{t,1}}} = \frac{T-\frac{2T}{3}}{T} = \frac{1}{3},
\end{eqnarray*}
which means that the data access latency of AAKL (AAUL) algorithm is reduced by $33.3\%$ compared with DA algorithm.

 Assume that $X_1$ and $X_3$ follow Shifted-Exponential distribution (Shifted-Exp $(\lambda, s)$), and $X_2$
 is a constant $T$, then $Y_{t,1} = T$, $Y_{t,2} = \min(T, \max(X_1, X_3))$. 
 Under a similar analysis, when $\lambda = \frac{1}{50}$, $$\Gamma_{U}^{Worst} = \frac{E[Y_{t,1}]-E[Y_{t,2}]}{E_{Y_{t,1}}} = 0.3708,$$ when $\lambda = \frac{1}{25}$, $$\Gamma_{U}^{Worst} = \frac{E[Y_{t,1}]-E[Y_{t,2}]}{E_{Y_{t,1}}} = 0.635.$$
The data access latency of the AAKL (AAUL) algorithm is reduced by $37.08\%$ $(63.5\%)$ compared with DA algorithm.

\section{Simulations}

In this section, we use the Monte Carlo simulation method to evaluate the data
access latency for our algorithms and the existing DA method.

For uniform distribution (or Shifted-Exponential distribution), we set the data access latency $X_i$ of $n$ nodes to obey the uniform distribution on $[0,100]$ (or Shifted-Exp$(\lambda = 0.02, s=1)$), where $n = 10, k = 6$. In the experiment, we compared the actual latency of the DA algorithm and AAKL (or AAUL) algorithm in data access. We conducted a total of 10,000 independent random experiments, and obtained the results shown in Fig. \ref{fig.2.1} (or Fig. \ref{fig.2.2}).
Fig. \ref{fig.2.1} (or Fig. \ref{fig.2.2}) shows the curve of the absolute value of the data access latency of AAKL (or AAUL) and DA algorithm as the number of experiments changes. For uniform distribution, the average data access latency for AAKL (or AAUL) is 40.78 and the average data access latency for DA algorithm is 49.70 where $$\frac{49.70-40.78}{49.70} = 0.1794 \approx (1-c)^2 = (1-0.6)^2.$$ For Shifted-Exponential distribution, the average data access latency for AAKL (or AAUL) is 40.78 and the average data access latency for DA algorithm is 49.70 where $$\frac{51.07-30.37}{51.07} = 0.405  \approx 1-c = 0.4.$$ According to the experimental results, these two experiments basically verified our theoretical conclusions.
From the two figures, we can see that our AAKL (or AAUL) algorithm has a significant advantage over the DA algorithm, whether it is uniform distribution or Shifted-Exponential distribution.

\begin{figure}[htbp]
    \centering
    \includegraphics[width=0.48\textwidth]{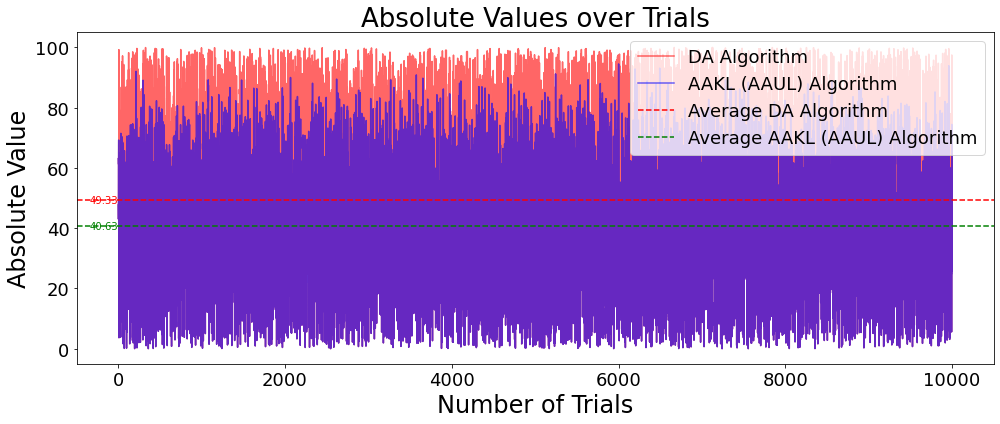}
    \caption{Data access latency for uniform distribution.}
    \label{fig.2.1}
\end{figure}

\begin{figure}[htbp]
    \centering
    \includegraphics[width=0.48\textwidth]{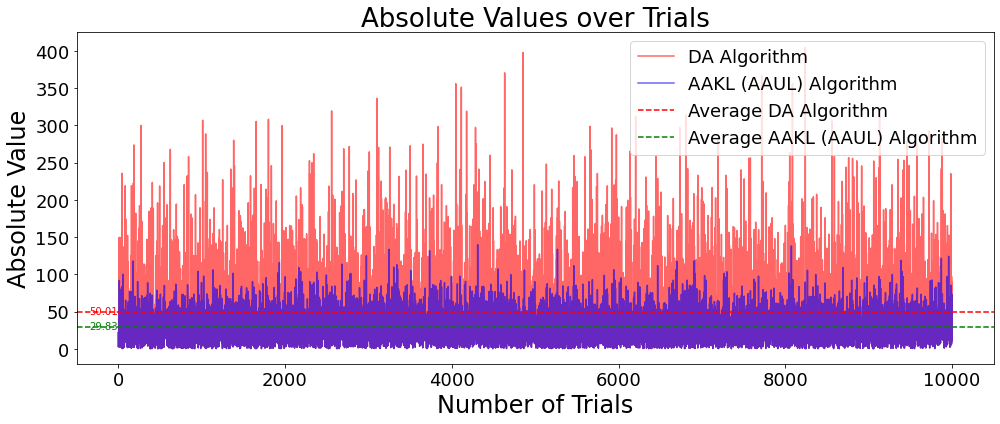}
    \caption{Data access latency for Shifted-Exponential distribution.}
    \label{fig.2.2}
\end{figure}

\section{Conclusion}
In this paper, we consider the problem of data download latency of a single node in a distributed storage system, 
and propose AAKL and AAUL algorithms based on MDS array codes. We show that our algorithms have less data access latency than the existing DA method. How to further reduce the data access latency for multiple nodes is one of our future work.

\ifCLASSOPTIONcaptionsoff
  \newpage
\fi

\bibliographystyle{IEEEtran}
\bibliography{main}

\end{document}